\documentclass[12pt,draftcls,onecolumn]{IEEEtran}

\usepackage{amssymb, bm, mathrsfs, color, subcaption}

\newtheorem{theorem}{Theorem}[section]

\newtheorem{assumption}[theorem]{Assumption}

\newtheorem{lemma}[theorem]{Lemma}
\newtheorem{remark}[theorem]{Remark}

\newcommand{\diag}{\mathop{\rm diag}\nolimits}

\usepackage{cite}

\ifCLASSINFOpdf
\usepackage[pdftex]{graphicx}
\else
\fi
\usepackage[cmex10]{amsmath}
\usepackage{array}
\usepackage{url}

\begin{document}
	%
	\title{Stabilization of Continuous-time 
		Switched Linear Systems with
		Quantized Output Feedback \thanks{
This technical note was partially presented at the 21st international
symposium on mathematical theory of networks and systems,
July 7-11, 2014, Netherlands.}
	}
	
	\author{Masashi~Wakaiki,~\IEEEmembership{Member,~IEEE,}
		and 
		Yutaka~Yamamoto,~\IEEEmembership{Fellow,~IEEE}
		\thanks{M. Wakaiki is with the Center for Control,
			Dynamical-systems and Computation (CCDC), University of California,
			Santa Barbara, CA 93106-9560, USA
			(e-mail:{\tt  \ masashiwakaiki@ece.ucsb.edu})
			Y. Yamamoto is with Department of Applied Analysis and Complex
			Dynamical Systems, Graduate School of Informatics, Kyoto University, Kyoto
			606-8501, Japan.
			.}
		\thanks{
			M. Wakaiki acknowledges Murata Overseas Scholarship Foundation
			for the support of this work.}
	}

	\maketitle
	
	\begin{abstract}
In this paper, we study the problem of stabilizing
continuous-time switched linear systems with 
quantized output feedback.
We assume that the observer and the control gain are given 
for each mode.
Also,
the plant mode is known to the controller and the quantizer.
Extending the result in the non-switched case,
we develop an update rule of the quantizer to achieve
asymptotic stability of the closed-loop system under
the average dwell-time assumption.
To avoid quantizer saturation,
we adjust the quantizer
at every switching time.

	\end{abstract}
	
	
	%
	\IEEEpeerreviewmaketitle

\section{Introduction}
Quantized control problems have been an active research topic
in the past two decades.
Discrete-level actuators/sensors and digital communication channels
are typical in practical control systems, and 
they yield quantized signals in feedback loops.
Quantization errors
lead to poor system performance and even loss of stability. Therefore,
various control techniques to explicitly
take quantization into account have been proposed,
as surveyed in \cite{Nair2007, Ishii2012}.

On the other hand, switched system models are widely used 
as a mathematical framework to represent both
continuous and discrete dynamics. For example,
such models are applied to
DC-DC converters \cite{Deaecto2010} and
to 
car engines \cite{Rinehart2008}.  
Stability and stabilization of switched systems
have also been extensively studied; see, e.g.,
the survey \cite{Shorten2007, Lin2009},
the book \cite{Liberzon2003Book}, and many references therein.

In view of the practical importance of both research areas
and common technical tools to study them,
the extension of
quantized control to switched systems has recently received 
increasing attention.
There is by now a stream of papers on
control with limited information for discrete-time 
Markovian jump systems \cite{Nair2003, Xiao2010, Xu2013}.
Moreover, our previous work \cite{Wakaiki2014IFAC} has analyzed 
the stability of sampled-data switched systems with static quantizers.

In this paper, we study the stabilization of continuous-time switched
linear systems with quantized output feedback.
Our objective is to solve the following problem:
Given a switched system and a controller, design
a quantizer to achieve asymptotic stability of the closed-loop
system. 
We assume that the information of the currently active plant mode
is available to the controller and the quantizer.
Extending the quantizer in \cite{Brockett2000, Liberzon2003Automatica}
for the non-switched case to the switched case,
we propose a Lyapunov-based update rule of the quantizer
under a slow-switching assumption of average dwell-time type 
\cite{Hespanha1999CDC}.

The difficulty of quantized control for switched systems
is that a mode switch changes the state trajectories and
saturates the quantizer.
In the non-switched case 
\cite{Brockett2000,Liberzon2003Automatica},
in order to avoid quantizer saturation, 
the quantizer is updated so that
the state trajectories always belong to 
certain invariant regions defined by
level sets of a Lyapunov function.
However, for switched systems, these invariant regions
are dependent on the modes. Hence 
the state may not belong to such regions after a switch.
To keep the state in the invariant regions, 
we here adjust the quantizer at every switching time, which
prevent quantizer saturation.

The same philosophy of emphasizing the importance of
quantizer updates after switching has been proposed
in \cite{Liberzon2014} for
sampled-data switched systems with quantized state feedback.
Subsequently, related works were presented 
for the output feedback case \cite{Wakaiki2014CDC} 
and for the case with bounded disturbances \cite{Yang2015ACC}.
The crucial difference lies in the fact that these works use
the quantizer based on \cite{Liberzon2003} and
investigates propagation of reachable sets 
for capturing the measurement.
This approach also aims to avoid quantizer saturation, but
it is fundamentally disparate from our Lyapunov-based approach.

This paper is organized as follows. In Section II, we 
present the main result, Theorem \ref{thm:stability_theorem},
after explaining the components 
of the closed-loop system.
Section III gives the update rule of the quantizer and 
is devoted to the proof of the convergence of the 
state to the origin.
In Section IV, we discuss Lyapunov stability.
We present a numerical example in Section V and finally conclude this paper 
in Section VI.

The present paper is based on the conference paper \cite{WakaikiMTNS2014}.
Here we extend the conference version by addressing
state jumps at switching times.
We also made structural improvements in this version.

{\em Notation:~}
Let $\lambda_{\min}(P)$ and $\lambda_{\max}(P)$ denote 
the smallest and the largest eigenvalue of $P \in \mathbb{R}^{\sf n\times n}$.
Let $M^{\top}$ denote the transpose of $M \in \mathbb{R}^{\sf m\times n}$.

The Euclidean norm of $v \in \mathbb{R}^{\sf n}$ is
denoted by $|v| = (v^*v)^{1/2}$. 
The Euclidean induced norm of $M \in \mathbb{R}^{\sf m\times n}$ is defined by
$\|M\| = \sup \{  |Mv |:~v\in \mathbb{R}^{\sf n},~|v|= 1 \}$.

For a piecewise continuous function $f:~\mathbb{R} \to \mathbb{R}$,
its left-sided limit at $t_0 \in \mathbb{R}$ is denoted by 
$f(t_0^{-}) = \lim_{t \nearrow t_0}f(t)$.

\section{Quantized output feedback stabilization of
	switched systems}

\subsection{Switched linear systems}
For a finite index set $\mathcal{P}$, let
$\sigma:[0,\infty) \to \mathcal{P}$ be a right-continuous and
piecewise constant function.
We call $\sigma$ a {\em switching signal} and
the discontinuities of $\sigma$ {\em switching times}.
Let us denote by $N_{\sigma}(t,s)$
the number of discontinuities of $\sigma$ on the interval $(s,t]$. 
Let $t_1,t_2,\dots$ be switching times, and
consider a switched linear system 
\begin{equation}
	\label{eq:SLS}
	\dot x(t) = A_{\sigma(t)}x(t)+B_{\sigma(t)}u(t),\quad 
	y(t) = C_{\sigma(t)}x(t)
\end{equation}
with the jump
\begin{equation}
	\label{eq:state_jump}
	x(t_k) = R_{\sigma(t_k),\sigma(t^-_k)} x(t^-_k)
\end{equation}
where $x(t) \in \mathbb{R}^{\sf{n}}$ is the state,
$u(t) \in \mathbb{R}^{\sf{m}}$ is the control input, and
$y(t) \in \mathbb{R}^{\sf{p}}$ is the output.

Assumptions on the switched system \eqref{eq:SLS} are as follows.
\begin{assumption}
	\label{ass:system}
	{\em
		For every $p \in \mathcal{P}$, $(A_p, B_p)$ is stabilizable and 
		$(C_p, A_p)$ is observable. 
		We choose $K_p \in \mathbb{R}^{\sf m \times n}$ and 
		$L_p \in \mathbb{R}^{\sf n \times p}$ so that
		$A_p+B_pK_p$ and $A_p+L_pC_p$ are Hurwitz.
		
		Furthermore, the switching signal $\sigma$ has an average dwell time~\cite{Hespanha1999CDC},
		i.e., there exist $\tau_a>0$ and $N_0 \geq 1$ such that
		\begin{equation}
			\label{eq:ADT_cond}
			N_{\sigma}(t,s) \leq N_0 + \frac{t-s}{\tau_a}
			\qquad (t > s \geq 0).
		\end{equation}
	}
\end{assumption}

We need observability rather than detectability, because
we reconstruct the state by using the observability Gramian.

\subsection{Quantizer}
In this paper, 
we use the following class of quantizers proposed in \cite{Liberzon2003Automatica}.

Let $\mathcal{Q}$ be a finite subset of $\mathbb{R}^{\sf{p}}$.
A quantizer is a piecewise constant function 
$q:\mathbb{R}^{\sf{p}} \to \mathcal{Q}$.
This implies geometrically that
$\mathbb{R^{\sf{p}}}$ is divided into a finite number of
the quantization regions
$\{y \in \mathbb{R^{\sf{p}}}:~q(y) = y_i \}$ $(y_i \in \mathcal{Q})$.
For the quantizer $q$, there exist positive numbers $M$ and $\Delta$ with
$M > \Delta$ such that
\begin{align}
	\label{eq:quantizer_cond1_nonSaturation}
	|y| \leq M &\quad \Rightarrow \quad |q(y) - y| \leq \Delta \\
	\label{eq:quantizer_cond2}
	|y| > M &\quad \Rightarrow \quad |q(y) | > M - \Delta.
\end{align}
The former condition \eqref{eq:quantizer_cond1_nonSaturation} gives
an upper bound of the quantization error when the quantizer 
does not saturate.
The latter \eqref{eq:quantizer_cond2} is used for 
the detection of quantizer saturation.

We place the following assumption on the 
behavior of the quantizer near the origin.
This assumption is used for Lyapunov stability of the closed-loop system. 
\begin{assumption}
	[\hspace{-0.01pt}\cite{Liberzon2003Automatica,Liberzon2007}]
	\label{ass:near origin}
	{\em
		There exists $\Delta_0 > 0$ such that $q(y)=0$
		for every $y \in \mathbb{R}^{\sf{p}}$ with $|y|\leq \Delta_0$.
	}
\end{assumption}

We use quantizers 
with the following adjustable parameter $\mu > 0$:
\begin{equation}
	\label{eq:zoom_q}
	q_{\mu}(y) = \mu q\left( \frac{y}{\mu}\right).
\end{equation}
In \eqref{eq:zoom_q}, $\mu$ is regarded as a ``zoom'' variable,
and $q_{\mu(t)}(y(t))$ is the data on $y(t)$ transmitted to the controller
at time $t$.
We need to change $\mu$ to obtain accurate information of $y$. 
The reader can refer to 
\cite{Liberzon2003Automatica, Liberzon2007, Liberzon2003Book}
for further discussions.

\begin{remark}
	The quantized output $q_{\mu}(y)$ may chatter on boundaries
	among quantization regions. Hence if we generate 
	the input $u$ by $q_{\mu}(y)$,
	the solutions of \eqref{eq:SLS} must be interpreted 
	in the sense of 
	Filippov~\cite{Filippov1988}.
	However, this generalization does not affect our Lyapunov-based analysis
	as in \cite{Brockett2000, Liberzon2003Automatica},
	because we will use a single quadratic Lyapunov function between switching times.
\end{remark}

\subsection{Controller}
Similarly to \cite{Brockett2000, Liberzon2003Automatica},
we construct the following dynamic output feedback law based on 
the standard Luenberger observers:
\begin{align}
	\dot \xi(t) &= (A_{\sigma(t)}+L_{\sigma(t)}C_{\sigma(t)})\xi(t) + 
	B_{\sigma(t)}u(t) - L_{\sigma(t)}q_{\mu(t)}(y(t)) \notag \\
	u(t) &= K_{\sigma(t)} \xi(t),
	\label{eq:controller_def}
\end{align}
where $\xi(t) \in \mathbb{R}^{\sf{n}}$ is the state estimate.
The estimate also jumps at each switching times $t_k$:
\begin{equation*}
	\xi(t_k) = R_{\sigma(t_k),\sigma(t^-_k)} \xi(t^-_k).
\end{equation*}
Then the closed-loop system is given by
\begin{align}
	\label{eq:ClosedSystem}
	\begin{array}{l}
		\dot x = A_{\sigma}x + B_{\sigma}K_{\sigma}\xi \\
		\dot \xi = (A_{\sigma}+L_{\sigma}C_{\sigma})\xi +
		B_{\sigma}K_{\sigma}\xi - L_{\sigma} q_{\mu}(y).
	\end{array}
\end{align}
If we define $z$ and $F_{\sigma}$ by
\begin{equation}
	z := 
	\begin{bmatrix}
		x \\ x - \xi
	\end{bmatrix},\quad
	F_{\sigma} := 
	\begin{bmatrix}
		A_{\sigma}+B_{\sigma}K_{\sigma} & -B_{\sigma}K_{\sigma} \\
		0 & A_{\sigma}+L_{\sigma}C_{\sigma}
	\end{bmatrix},
	\label{eq:F_def}
\end{equation}
then we rewrite \eqref{eq:ClosedSystem} in the form
\begin{equation}
	\label{eq:ClosedSystem_simple}
	\dot z = F_{\sigma} z +
	\begin{bmatrix} 
		0 \\ L_{\sigma}
	\end{bmatrix}
	(q_{\mu}(y) - y).
\end{equation}
The state $z$ of the closed-loop system \eqref{eq:ClosedSystem}
jumps at each switching time $t_k$:
\begin{equation*}
	z(t_k) = J_{\sigma(t_k),\sigma(t^-_k)} z(t^-_k),
\end{equation*}
where
\begin{equation*}
	J_{\sigma(t_k),\sigma(t^-_k)} :=
	\begin{bmatrix}
		R_{\sigma(t_k),\sigma(t^-_k)} & 0 \\
		0 & R_{\sigma(t_k),\sigma(t^-_k)}
	\end{bmatrix}.
\end{equation*}
We see from Assumption \ref{ass:system} that 
$F_p$ is Hurwitz for each $p \in \mathcal{P}$.
For every positive-definite matrix 
$Q_p \in \mathbb{R}^{2{\sf n} \rm \times 2 {\sf n}}$,
there exist a positive-definite matrix
$P_p \in \mathbb{R}^{2{\sf n} \rm \times 2 {\sf n}}$ such that
\begin{equation}
	\label{eq:Lyapunov_cont}
	F_p^{\top}P_p + P_pF_p = -Q_p
	\qquad (p \in \mathcal{P}).
\end{equation}
We define 
$\overline \lambda_P$, $\underline \lambda_P$, 
$\underline \lambda_Q$, and $C_{\max}$ by
\begin{align}
	\label{eq:alpha_def}
	\begin{array}{c}
		\overline \lambda_P := 
		{\displaystyle \max_{p \in \mathcal{P}} \lambda_{\max}(P_p)},
		\quad
		\underline \lambda_P  := 
		{\displaystyle \min_{p \in \mathcal{P}} \lambda_{\min}(P_p)} \\[8pt]
		\underline \lambda_Q  := 
		{\displaystyle \min_{p \in \mathcal{P}} \lambda_{\min}(Q_p)},
		\quad
		C_{\max} := 
		{\displaystyle \max_{p \in \mathcal{P}}\|C_p\|}.
	\end{array}
\end{align}
Fig.~\ref{fig:CSSQOF} shows the closed-loop system we consider.

\subsection{Main result}
By adjusting the ``zoom'' parameter $\mu$,
we can achieve global asymptotic stability of 
the closed-loop system \eqref{eq:ClosedSystem_simple}.
This result is a natural extension of Theorem 5
in \cite{Liberzon2003Automatica}
to switched systems.

\begin{figure}[bt]
	\centering
	\includegraphics[width = 9cm,clip]
	{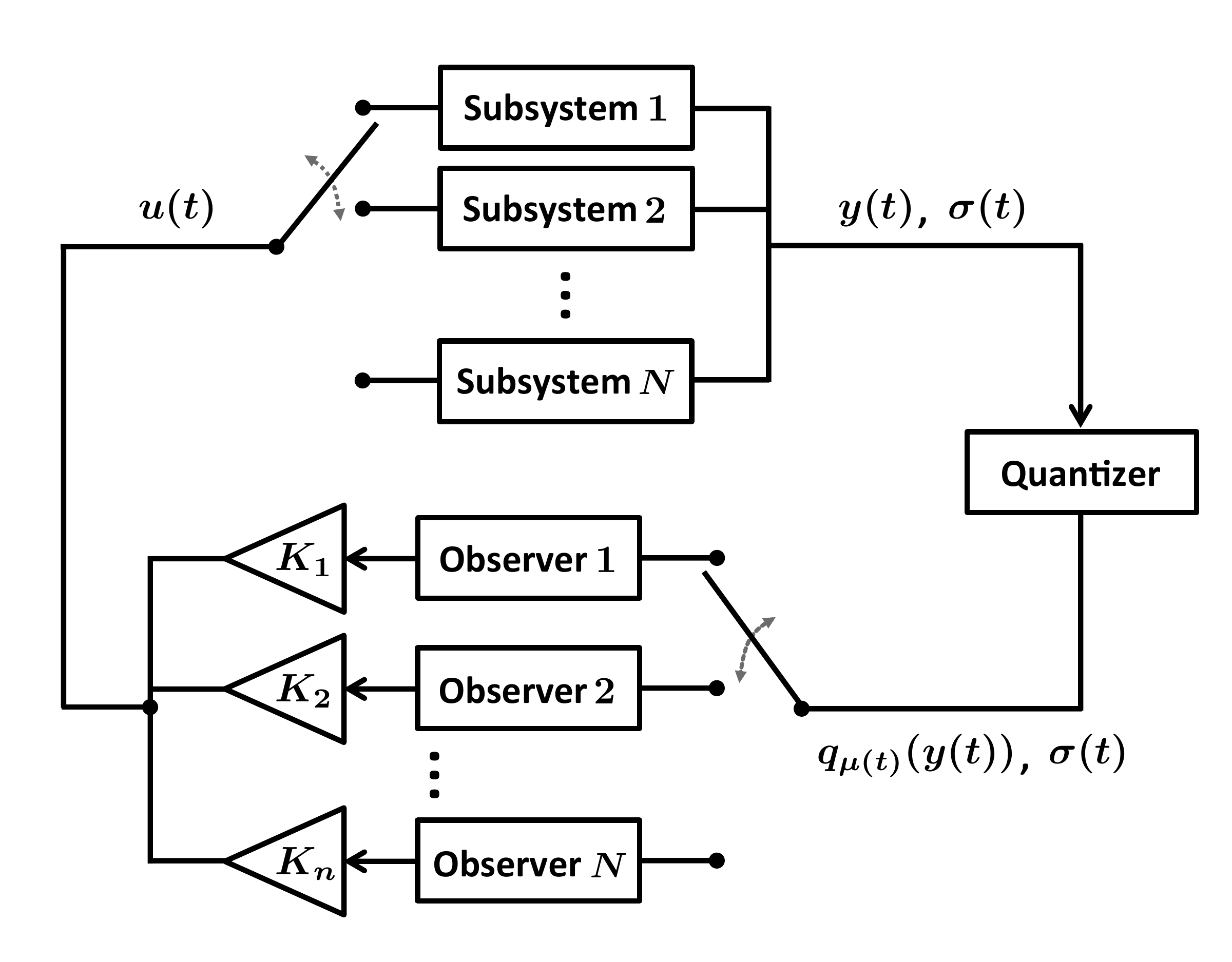}
	\caption{Continuous-time switched system with quantized output feedback.}
	\label{fig:CSSQOF}
\end{figure}

\begin{theorem}
	{\em
		\label{thm:stability_theorem}
		Define $\Theta$ by
		\begin{equation}
			\label{eq:Theta_def}
			\Theta := \frac{2 \max_{p\in \mathcal{P}} \|P_p\hat L_p\|}{\underline \lambda_Q},
			\text{~~where~~}
			\hat L_p := 
			\begin{bmatrix} 
				0 \\ L_{p}
			\end{bmatrix}.
		\end{equation}
		and let $M$ be large enough to satisfy
		\begin{equation}
			\label{eq:M_Delta_cond_main_thm}
			M > \max
			\left\{ 2\Delta,\quad
			\sqrt{\frac{\overline \lambda_P}{\underline \lambda_P}}
			\Theta \Delta C_{\max}
			\right\}.
		\end{equation}
		If the average dwell time $\tau_a$ in
		\eqref{eq:ADT_cond} 
		is larger than a certain value, then
		there exists a right-continuous and piecewise-constant function $\mu$ such that
		the closed-loop system \eqref{eq:ClosedSystem_simple} has
		the following two properties 
		for every $x(0) \in \mathbb{R}^{\sf{n}}$ and every $\sigma(0) \in \mathcal{P}$:
		
		{\sl (i)~Convergence to the origin:~}
		$\lim_{t \to \infty}z(t) = 0$.
		
		{\sl (ii)~Lyapunov stability:~}
		To every $\varepsilon > 0$, there corresponds $\delta > 0$ such that
		\begin{equation*}
			|x(0)| < \delta \quad \Rightarrow \quad
			|z(t)| < \varepsilon~~~(t\geq 0).
		\end{equation*}
	}
\end{theorem}

We shall prove convergence to the origin and 
Lyapunov stability in Sections III and IV, respectively.
We also present an update rule of the ``zoom'' parameter $\mu$
in Section 3.
The sufficient condition on $\tau_a$ is given by 
\eqref{eq:adt_cond} in Theorem \ref{lem:convergence_to_origin} below.

\section{The proof of convergence to the origin}
Define $\Gamma$ and $\Lambda$ by
\[
\Gamma := \max_{p \in \mathcal{P}}\|A_p\|,\quad
\Lambda := \max
\left\{
1,~\max_{p,q \in \mathcal{P},p\not=q }\|R_{p,q}\|
\right\}.
\]
We split the proof into two stages:
the ``zooming-out'' and ``zooming-in'' stages.

\subsection{Capturing the state of the closed-loop system 
	by ``zooming out''} 
Since the initial state $x(0)$ is unknown to the quantizer,
we have to capture 
the state $z$ of the closed-loop
system
by ``zooming out'', i.e., increasing the ``zoom'' parameter $\mu$.
We first see that $z$ can be captured if
we have a time-interval with a given length that has
no switches.
\begin{theorem}
	\label{lem:capture_non_switched}
	{\em
		Consider the closed-loop system \eqref{eq:ClosedSystem_simple}.
		Set the control input $u = 0$.
		Choose $\tau > 0$, and define
		$\Upsilon_p(\tau):= \max_{0\leq t \leq \tau} \left\| C_pe^{A_pt} 
		\right\|$ and the observability Gramian
		\[
		W_p(\tau) :=
		\int^{\tau}_0 e^{A_p^{\top}t}C_p^{\top} C_pe^{A_pt} dt.
		\]
		Assume that there exists $s_0 \geq 0$ such that
		we can observe
		\begin{gather}
			|q_{\mu(t)}(y(t))| \leq 
			M \mu(t) - \Delta \mu(t) \label{eq:zooming_out_QY} \\
			\sigma(t) = \sigma(s_0) =: p \label{eq:zooming_out_SS}
		\end{gather}
		for all $t \in [s_0, s_0 + \tau)$.
		Let the ``zoom'' parameter $\mu$ be
		piecewise continuous and monotone increasing
		in $[0,s_0+\tau)$.
		If we set the state estimate $\xi$ at $t = s_0+\tau$ by
		\begin{equation}
			\label{eq:xi_t0_tau}
			\xi(s_0+\tau) 
			:=
			e^{A_p \tau}
			\left(
			W_p(\tau)^{-1}
			\int^{\tau}_{0}  e^{A_p^{\top}t}C_p^{\top} q_{\mu(s_0+t)}(y(s_0+t))dt
			\right) 
		\end{equation}
		and if we choose $\mu(s_0+\tau)$ so that
		\begin{align}
			\mu(s_0+\tau) \geq
			\sqrt{\frac{\overline \lambda_p}{\underline \lambda_p}}
			\frac{C_{\max}}{M}
			\biggl(
			|\xi(s_0+\tau)| 
			+
			2\|W_p(\tau)^{-1}\|\tau \Upsilon_p(\tau) \left\| e^{A_p\tau} \right\| 
			\Delta \mu((s_0+\tau)^-)
			\biggr),
			\label{eq:mu_s_0_tau}
		\end{align}
		then $z(s_0 + \tau)\in 
		\mathscr{R}_1(\mu(s_0 + \tau),\sigma(s_0 + \tau))$.
	}
\end{theorem}
\begin{proof}
	Since no switch occurs by \eqref{eq:zooming_out_SS},
	we can easily obtain this result by extending Theorem 5 in 
	\cite{Liberzon2003Automatica} for the non-switched case. 
	We therefore omit the proof;
	see also the conference version \cite{WakaikiMTNS2014}.
\end{proof}

It follows from Theorem \ref{lem:capture_non_switched} that
in order to capture the state $z$,
it is enough
to show the existence of $s_0 \geq 0$ satisfying 
\eqref{eq:zooming_out_QY} and \eqref{eq:zooming_out_SS} for all
$t \in [s_0,s_0 + \tau)$.
To this end, we use the following lemma on 
average dwell time $\tau_a$:
\begin{lemma}
	\label{lem:ADT_upperbound}
	{\em
		Fix an initial time $t_0 \geq 0$.
		Suppose that $\sigma$ satisfies 
		the average dwell-time assumption \eqref{eq:ADT_cond}.
		Let $\tau \in (0, \tau_a)$. If we choose $N\in \mathbb{N}$ so that
		\begin{equation}
			\label{eq:N_ADTcond}
			N > \frac{\tau_a}{\tau_a - \tau} \left( N_0 - \frac{\tau}{\tau_a} \right),
		\end{equation}
		then there exists  $\upsilon \in [0, (N-1)\tau]$ such that 
		$N_{\sigma}(t_0+\upsilon+\tau,t_0+\upsilon) = 0$.
	}
\end{lemma}
\begin{proof}
	Let us denote the switching times by $t_1,t_2,\dots$, and
	fix $N\in \mathbb{N}$.
	Suppose that
	\begin{equation}
		\label{eq:N_sigma>0}
		N_{\sigma}(t_0+\upsilon+\tau,t_0+\upsilon) > 0
	\end{equation}
	for all $\upsilon \in [0, (N-1)\tau]$. 
	Then
	we have 
	\begin{equation}
		\label{eq:t_k_tau_bound}
		t_k - t_{k-1} \leq \tau \qquad (k=1,\dots,N).
	\end{equation}
	Indeed, if 
	$t_k - t_{k-1} > \tau
	$
	for some $k \leq N$ and 
	if we let $\bar k$ be the smallest such integer,
	then we obtain 
	\[
	t_{\bar k -1} - t_0 \leq (\bar k -1)\tau
	\leq (N-1)\tau
	\] 
	and
	$N_{\sigma}(t_{\bar k-1}+\tau , t_{\bar k-1}) = 0$. This contradicts
	\eqref{eq:N_sigma>0}
	with $\upsilon = t_{\bar k-1} - t_0 \in [0,(N-1)\tau]$.
	Thus we have \eqref{eq:t_k_tau_bound}.
	
	From \eqref{eq:t_k_tau_bound}, we see that
	for $0 < \epsilon < t_1$,
	\begin{equation*}
		t_N- (t_1 - \epsilon) = \sum_{k=2}^N (t_k -t_{k-1}) + \epsilon
		\leq (N-1)\tau + \epsilon
	\end{equation*}
	It follows from \eqref{eq:ADT_cond} that
	\begin{align*}
		N = N_{\sigma}(t_N,t_1 - \epsilon) 
		&\leq N_0 + \frac{(N-1)\tau + \epsilon}{\tau_a}.
	\end{align*}
	Therefore $N$ satisfies the following inequality:
	\begin{equation}
		\label{eq:n_upperbound}
		N \leq \frac{\tau_a}{\tau_a - \tau} 
		\left(
		N_0 - \frac{\tau - \epsilon}{\tau_a}
		\right).
	\end{equation}
	Since $\epsilon \in (0, t_1)$ was arbitrary, 
	\eqref{eq:n_upperbound} is equivalent to
	\begin{equation}
		\label{eq:n_upperbound2}
		N \leq \frac{\tau_a}{\tau_a - \tau} 
		\left(
		N_0 - \frac{\tau }{\tau_a}
		\right).
	\end{equation}
	Thus we have shown that
	if \eqref{eq:N_sigma>0} holds for all $\upsilon \in [0, (N-1)\tau]$,
	then $N \in \mathbb{N}$ satisfies \eqref{eq:n_upperbound2}.
	The contraposition of this statement gives a desired result. 
\end{proof}

\begin{theorem}
	\label{thm:non_switch_interval}
	{\em
		Consider the closed-loop system \eqref{eq:ClosedSystem_simple}
		with average dwell-time property \eqref{eq:ADT_cond}.
		Set the control input $u = 0$. Fix $\chi > 0$, $\bar \tau > 0$, 
		and $\tau \in (0, \tau_a)$.
		Increase $\mu$ in the following way:
		$\mu(t)=1$ for $t \in [0,\bar \tau)$, 
		\begin{equation}
			\label{eq:mu_update_zoom_in}
			\mu(t)= \Lambda^{
				N_0
			} \cdot
			\left(
			\Lambda^{1/\tau_a} e^{\Gamma}
			\right)^{(1+\chi)k \bar \tau}
		\end{equation}
		for $t \in [k \bar \tau,(k+1)\bar \tau)$ and $k \in \mathbb{N}$.
		Then there exists
		$s_0 \geq 0$ such that
		\eqref{eq:zooming_out_QY} and \eqref{eq:zooming_out_SS} hold for all
		$t \in [s_0,s_0 + \tau)$.
	}
\end{theorem}
\begin{proof}
	If $n $ switches occur in the interval $(0,t]$, then we have
	\begin{equation*}
		| x(t) |
		\leq 
		\left(
		\prod_{k=1}^n \Lambda 
		\right) \cdot
		e^{\Gamma t} \cdot |x(0)|.
	\end{equation*}
	Since $\Lambda \geq 1$, it follows from \eqref{eq:ADT_cond} that
	\begin{equation}
		\label{eq:state_bound}
		| x(t) |
		\leq 
		\Lambda^{\left(
			N_0 +
			\frac{t}{\tau_a}
			\right)} \cdot
		e^{\Gamma t} \cdot  |x(0)|.
	\end{equation}
	Clearly, this inequality holds in the case when
	no switches occur.
	Since \eqref{eq:M_Delta_cond_main_thm} shows that $M - 2\Delta > 0$ 
	and since the growth rate of $\mu(t)$ is larger than 
	that of $|y(t)|$,
	there exists $s_0' \geq 0$ such that
	\begin{equation}
		\label{eq:y_M}
		\left|
		y(t)
		\right| \leq M\mu(t) -2\Delta \mu(t)
		\qquad (t \geq s_0').
	\end{equation}
	In conjunction with \eqref{eq:quantizer_cond1_nonSaturation},
	this implies that \eqref{eq:zooming_out_QY} holds for 
	every $t \geq s_0'$.
	Let $N$ be an integer satisfying \eqref{eq:N_ADTcond}.
	Then Lemma \ref{lem:ADT_upperbound} guarantees the
	existence of $s_0 \in [s_0', s_0' + (N-1)\tau]$ such that 
	\eqref{eq:zooming_out_SS} holds for every $t \in [s_0, s_0 + \tau)$.
	This completes the proof.
\end{proof}

It follows from Theorems \ref{lem:capture_non_switched} and 
\ref{thm:non_switch_interval} that
if we update the ``zoom'' 
parameter $\mu$ as in \eqref{eq:mu_update_zoom_in}
and if we set the state estimate $\xi$ by \eqref{eq:xi_t0_tau},
then the state $z$ of the closed-loop system
can be captured.

\begin{remark}
	\label{rem:brief_remark_for_Lyap}
	If the initial state $x(0)$ is sufficiently small, then
	$s_0'$ in \eqref{eq:y_M} is zero. In this situation, we can capture 
	$z$ by $t = N\tau$ for all switching signal with average
	dwell-time property \eqref{eq:ADT_cond}. We use this fact 
	for the proof of Lyapunov stability; see Section 4.
\end{remark}

\subsection{Measuring the output by ``zooming in''}
Next we drive the state $z$ of the closed-loop system
to the origin by ``zooming-in'', i.e., 
decreasing the ``zoom'' parameter $\mu$.
Since
$\mu$ increases
at each switching time during this stage,
the term ``zooming-in stage'' may be misleading.
However, $\mu$ decreases overall under a certain
average dwell-time assumption \eqref{eq:ADT_cond}, 
so we use the term ``zooming-in'' 
as in \cite{
	Brockett2000,
	Liberzon2003Automatica}.

Let us first consider a fixed ``zoom'' parameter $\mu$.
The following lemma shows that
if no switches occur, then the state trajectories move
from a large level set to a small level set
of the Lyapunov function $V_p(z) := z^{\top}P_pz$
in a finite time that is independent of the mode $p$:
\begin{lemma}
	\label{lem:fix_zoom_parameter}
	{\em
		Define $F_p$ and $\hat L_p$ as in \eqref{eq:F_def} and
		\eqref{eq:Theta_def}, respectively.
		Fix $p \in \mathcal{P}$, and consider
		the non-switched system
		\begin{equation}
			\label{eq:noSwitchSystem}
			\dot z = F_p z + \hat L_p 
			(q_{\mu}(y) - y).
		\end{equation}
		Choose $\kappa > 0$. If 
		$M$ satisfies
		\begin{equation}
			\label{eq:M_Delta_cond}
			\sqrt{\underline \lambda_P} M > \sqrt{\overline \lambda_P} \Theta \Delta(1+ \kappa)
			C_{\max},
		\end{equation}
		where 
		$\overline \lambda_P$, 
		$\underline \lambda_P$
		$C_{\max}$, and $\Theta$ are defined by 
		\eqref{eq:alpha_def} and
		\eqref{eq:Theta_def},
		then the following two level sets of the Lyapunov function
		$V_p(z):= z^{\top}P_p z$
		are invariant regions for every trajectory of \eqref{eq:noSwitchSystem}:
		\begin{align}
			\mathscr{R}_1 (\mu,p) &:= 
			\left\{
			z\in \mathbb{R}^{\sf n}:V_p(z) \leq 
			\frac{\underline \lambda_P M^2\mu^2}{C_{\max}^2}
			\right\}
			\label{eq:R1_def} \\
			\mathscr{R}_2 (\mu,p) &:= 
			\left\{
			z\in \mathbb{R}^{\sf n}:V_p(z) \leq 
			\overline \lambda_P (\Theta \Delta (1+\kappa))^2\mu^2
			\right\}.
			\label{eq:R2_def}
		\end{align}
		Furthermore, if $z(t) \in \mathscr{R}_1 (\mu,p) \setminus 
		\mathscr{R}_2 (\mu,p)$ for all $t \in [t_1,t_2]$, then
		\begin{equation}
			\label{eq:Lyapunov_decrease}
			V_p(z(t_2)) \leq
			V_p(z(t_1)) - 
			(t_2 - t_1) \underline \lambda_Q
			\kappa (1+\kappa)(\Theta \Delta \mu)^2
		\end{equation}
		for every $p \in \mathcal{P}$.
		Hence if $T$ satisfies
		\begin{equation}
			\label{eq:T_def}
			T >
			\frac{\underline \lambda_P M^2 - \overline \lambda_P (\Theta 
				\Delta(1+\kappa) C_{\max})^2}
			{\underline \lambda_Q \kappa(1+\kappa)(\Theta \Delta C_{\max})^2},
		\end{equation}
		then every trajectory of \eqref{eq:noSwitchSystem} with an
		initial state $z(0) \in \mathscr{R}_1 (\mu,p)$ satisfies
		$z(T) \in \mathscr{R}_2 (\mu,p)$．
	}
\end{lemma}

\begin{proof}
	Since the mode $p \in \mathcal{P}$ is fixed,
	this lemma is a trivial extension of Lemma~5 
	in \cite{Liberzon2003Automatica} for single-modal systems.
	We therefore omit its proof;
	see also the conference version \cite{WakaikiMTNS2014}.
\end{proof}

Using Lemma \ref{lem:fix_zoom_parameter},
we obtain an update rule of the ``zoom'' parameter $\mu$
to drive the state $z$ to the origin.
\begin{theorem}
	\label{lem:convergence_to_origin}
	{\em
		Consider the system \eqref{eq:noSwitchSystem} under
		the same assumptions as in Lemma \ref{lem:fix_zoom_parameter}.
		Assume that $z(t_0)\in \mathscr{R}_1(\mu(t_0),\sigma(t_0))$.
		For each $p_1,p_2 \in \mathcal{P}$ with
		$p_1 \not= p_2$, 
		the positive definite matrices $P_{p_1}$ and $P_{p_2}$ in 
		the Lyapunov equation \eqref{eq:Lyapunov_cont} satisfy
		\begin{equation}
			\label{eq:Lyapunov_cond}
			z^{\top} J_{p_2,p_1}^{\top}
			P_{p_2}
			J_{p_2,p_1}
			z \leq  c_{p_2,p_1} \cdot z^{\top} P_{p_1} z
			\qquad (z \in \mathbb{R}^{2{\sf n}})
		\end{equation}
		for some $c_{p_2,p_1} > 0$. Define $c$ and $\Omega$ by
		\begin{equation}
			\label{eq:c_def_cont}
			c := 
			\max
			\left\{
			1,~ \max_{p_1,p_2 \in \mathcal{P}, p_1 \not= p_2}
			c_{p_2,p_1}
			\right\}
		\end{equation}
		\begin{equation}
			\label{eq:Omega_def}
			\Omega := \sqrt{\frac{\overline \lambda_P}{\underline \lambda_P}}
			\frac{\Theta\Delta (1+\kappa)C_{\max}}{M} < 1.
		\end{equation}
		Fix $T > 0$ so that \eqref{eq:T_def} is satisfied, and
		set the ``zoom'' parameter $\mu(t_0+kT+t)$ for all $k \in \mathbb{Z}$ and
		$t \in (0,T]$ in the following way:
		If no switches occur in the interval 
		$(t_0+kT, t_0+(k+1)T]$, then
		\begin{align}
			\label{eq:mu_update_non_switched}
			\mu(t_0+kT+t) =
			\begin{cases}
				\mu(t_0+kT) & (0 < t < T) \\
				\Omega \mu(t_0)
				& (t = T);
			\end{cases}
		\end{align}
		otherwise,
		\begin{align}
			\label{eq:mu_update_switched}
			\mu(t_0+kT+t) =
			\begin{cases}
				\mu(t_0+kT) & (0 < t < t_1) \\
				\sqrt{
					\prod_{\ell = 0}^{i-1}c_{\sigma(t_{\ell + 1}), \sigma(t_{\ell})}} 
				\cdot \mu(t_0) &  (t_i \leq t < t_{i+1},~i = 1,\dots,n)  \\
				\Omega 	\prod_{\ell = 0}^{n-1}
				c_{\sigma(t_{\ell + 1}), \sigma(t_{\ell})} \cdot \mu(t_0)
				&  (t = T),
			\end{cases}
		\end{align}
		where $t_1,\dots, t_n$ are the switching times in the interval 
		$(t_0+kT, t_0+(k+1)T]$.
		Then $z(t) \in \mathscr{R}_1(\mu(t),\sigma(t))$ for all
		$t \geq t_0$. Furthermore, 
		if $\tau_a$ satisfies
		\begin{equation}
			\label{eq:adt_cond}
			\tau_a > \frac{\log(c)}{2\log(1/\Omega)} T,
		\end{equation}
		then $\lim_{t \to \infty} z(t) = 0$.
	}
\end{theorem}
\begin{proof}
	To prove that $z(t) \in \mathscr{R}_1(\mu(t),\sigma(t))$
	for all $t \geq t_0$,
	it is enough to show that if
	$z(t_0) \in \mathscr{R}_1(\mu(t_0),\sigma(t_0))$, then
	\begin{equation}
		\label{eq:z_t0_t0T}
		z(t) \in \mathscr{R}_1(\mu(t),\sigma(t)) \qquad
		(t_0 \leq t \leq t_0+T)
	\end{equation}
	
	Let us first investigate the case 
	without switching on the interval $(t_0, t_0+T]$.
	We see from
	Lemma \ref{lem:fix_zoom_parameter} that
	$z(t) \in \mathscr{R}_1(\mu(t),\sigma(t))$ for all
	$t \in [t_0, t_0+T)$ and that
	$z((t_0 + T)^-)\in \mathscr{R}_2(\mu(t_0),\sigma(t_0))$.
	Since $\mu(t_0+T) = \Omega \mu(t_0)$, 
	a routine calculation shows that
	$z(t_0+T) \in \mathscr{R}_1(\mu(t_0+T),\sigma(t_0+T))$.
	
	We now study the switched case. 
	Let $t_1,t_2,\dots,t_n$ be the switching times in the interval $(t_0, t_0+T]$.
	Let us define $t_{n+1} := t_0+T$ for simplicity of notation.
	Lemma \ref{lem:fix_zoom_parameter} implies that
	$\mathscr{R}_i (\mu(t_k),\sigma(t_k))$ ($i=1,2$)
	are invariant sets for all $t \in [t_k,t_{k+1})$, $k=0,\dots,n$.
	Moreover, by \eqref{eq:Lyapunov_cond},
	if $z(t_k^-) \in \mathscr{R}_i (\mu(t_k^-),\sigma(t_k^-))$, 
	then
	$z(t_k) \in \mathscr{R}_i 
	(\mu(t_k),\sigma(t_k))$ 
	($i=1,2$) for all $k = 1,\dots, n$.
	Hence $z(t_0) \in \mathscr{R}_1 (\mu(t_0),\sigma(t_0))$
	leads to
	\begin{equation}
		\label{eq:z_before_tn1}
		z(t) \in \mathscr{R}_1 (\mu(t),\sigma(t)) \qquad
		(t_0 \leq t < t_{n+1}).
	\end{equation}
	
	To obtain 
	\begin{equation}
		\label{eq:z_tn1}
		z(t_{n+1}) \in \mathscr{R}_1 (\mu(t_{n+1}),\sigma(t_{n+1})),
	\end{equation}
	we show that $z(t_{n+1}^-) \in \mathscr{R}_2 
	(\mu(t_{n+1}^-),\sigma(t_{n+1}^-))$. 
	Assume, to reach a contradiction, that 
	\begin{equation}
		\label{eq:z_NOt_R2}
		z(t_{n+1}^-) \not\in \mathscr{R}_2 
		(\mu(t_{n+1}^-),\sigma(t_{n+1}^-)).
	\end{equation}
	Since $\mathscr{R}_2 
	(\mu(t),\sigma(t))$ is an invariant region for all
	$t \in [t_0, t_{n+1})$, we also have 
	\begin{equation*}
		z(t) \not\in \mathscr{R}_2 
		(\mu(t),\sigma(t))
		\qquad(t_0 \leq t < t_{n+1}).
	\end{equation*}
	Define a Lyapunov function $V_p(z) := z^{\top}P_pz$ for
	each $p \in \mathcal{P}$.
	Since a Filippov solution is (absolutely) continuous,
	$\lim_{t \nearrow t_{k}} V_{\sigma(t)}(z(t))$ exists 
	for each $k=1,\dots,n+1$. From
	\eqref{eq:z_NOt_R2}, we obtain
	\begin{equation}
		\label{eq:Lyapunov_t_0+T}
		\lim_{t \nearrow t_{n+1}} V_{\sigma(t)}(z(t))
		\geq \overline \lambda_P (\Theta \Delta (1+\kappa))^2\mu(t_n)^2.
	\end{equation}
	
	On the other hand, since 
	$z(t) \in \mathscr{R}_1 (\mu(t),\sigma(t)) \setminus
	\mathscr{R}_2 (\mu(t),\sigma(t))$ for
	all $t \in [t_0, t_{1}]$, 
	\eqref{eq:Lyapunov_decrease} gives
	\begin{align*}
		\lim_{t \nearrow t_1}V_{\sigma(t)}(z(t))
		\!\leq \!
		\left(
		\frac{\underline \lambda_P M^2}{C_{\max}^2}
		\!-\!
		(t_1\!-\!t_0)\underline \lambda_Q \kappa(1\!+\!\kappa)(\Theta \Delta)^2
		\right)\mu(t_0)^2,
	\end{align*}
	and hence we have from 
	$\mu(t_1) = \sqrt{
		c_{\sigma(t_1),\sigma(t_0)}
	}
	\mu(t_0)	
	$
	that
	\begin{align*}
		V_{\sigma(t_1)}(z(t_1))
		&=
		z(t_1^-)^{\top}J_{\sigma(t_1), \sigma(t_1^-)}^{\top}
		P_{\sigma({t_1})}J_{\sigma(t_1), \sigma(t_1^-)}z(t_1^-) \\
		&\leq
		c_{\sigma(t_1), \sigma(t_0)} 
		\cdot \left( \lim_{t \nearrow t_1}V_{\sigma(t)}(z(t))\right) \\
		&=
		\left(
		\frac{\underline \lambda_P M^2}{C_{\max}^2}
		\!-\!
		(t_1\!-\!t_0)\underline \lambda_Q \kappa(1\!+\!\kappa)(\Theta \Delta)^2
		\right)
		\!\mu(t_1)^2.
	\end{align*}
	
	If we repeat this process and use \eqref{eq:T_def},
	then
	\begin{align}
		\lim_{t \nearrow t_{n+1}}
		V_{\sigma(t)}(z(t)) 
		&\leq
		\left(
		\frac{\underline \lambda_P M^2}{C_{\max}^2}
		-
		T \underline \lambda_Q \kappa (1+\kappa)(\Theta \Delta )^2
		\right) \mu(t_n)^2 \notag \\
		&<
		\overline \lambda_P (\Theta \Delta (1+\kappa))^2
		\mu(t_n)^2,
		\label{eq:Lyapunov_t_0+t_contradiction}
	\end{align}
	which contradicts \eqref{eq:Lyapunov_t_0+T}.
	Thus we obtain 
	\[z(t_{n+1}^-) \in \mathscr{R}_2 (\mu(t_{n+1}^-),\sigma(t_{n+1}^-)),\]
	and hence
	\eqref{eq:z_tn1} holds.
	
	From \eqref{eq:z_before_tn1} and \eqref{eq:z_tn1},
	we derive the desired result \eqref{eq:z_t0_t0T}, because
	$t_{n+1} = t_0+T$.
	
	Finally, 
	since $c \geq 1$, 
	\eqref{eq:ADT_cond} gives
	\begin{align}
		\mu(t_0+mT+t) 
		\leq \Omega^m \sqrt{c^{N_{\sigma}(t_0+mT+t, t_0)}}
		\mu(t_0) 
		\leq \sqrt{c^{N_0+T/\tau_a}} \cdot
		\left(\Omega \sqrt{c^{T/\tau_a}} \right)^m \mu(t_0)
		\label{eq:mu(t_0+MT)}
	\end{align}
	for every $m \geq 0$ and $t \in [0,T)$.
	If $\Omega \sqrt{c^{T/\tau_a}} < 1$, that is, if
	the average dwell time $\tau_a$ satisfies \eqref{eq:adt_cond},
	then $\lim_{t\to \infty}\mu(t) =0$.
	Since $z(t) \in \mathscr{R}_1(\mu(t),\sigma(t))$ for 
	all $t \geq t_0$, we obtain
	$\lim_{t \to \infty} z(t)= 0$.
\end{proof}


\begin{remark}
	\noindent
	{(a)~}
	We can compute $c_{p_2,p_1}$ by linear matrix inequalities.
	Moreover, if the jump matrix $R_{p_2,p_1}$ in \eqref{eq:state_jump}
	is invertible, 
	then Lemma 13 of \cite{Tanwani2011} gives 
	an explicit formula for $c_{p_2,p_1}$.
	
	\noindent
	{(b)~}
	The proposed method is
	sensitive to the time-delay of the switching signal at the ``zooming-in'' stage.
	If the switching signal is delayed,
	a mode mismatch occurs between the plant and the controller.
	Here we do not proceed along this line to avoid technical issues.
	See also \cite{Ma2015} for the stabilization of 
	asynchronous switched systems with time-delays.
	
	\noindent
	{(c)~}
	We have updated the ``zoom'' parameter $\mu$ 
	at each switching time in the ``zooming-in'' stage.
	If we would not, switching could lead to instability of the closed-loop system.
	In fact, 
	since the state $z$ may not belong to the invariant region
	$\mathscr{R}_1(\mu,\sigma)$ without adjusting $\mu$,
	the quantizer may saturate.
	
	\noindent
	{(d)~}
	Similarly,
	``pre-emptively'' multiplying $\mu$ at time $T_0+kT$ by $c^n$
	does not work, either.
	This is because such an adjustment does not make 
	$\mathscr{R}_1(\mu,\sigma)$ invariant for the state
	trajectories. For example,
	consider the situation where
	the state $z$ belongs to $\mathscr{R}_2(\mu,\sigma)$ at $t = T_0+kT$
	due to this pre-emptively adjustment.
	Then $z$ does not converge to the origin.
	Let
	$t_1 > T_0+kT$ be a switching time.
	Since 
	$\mathscr{R}_2(\mu(t_1^-),\sigma(t_1^-))$
	may not be a subset of
	$\mathscr{R}_1(\mu(t_1),\sigma(t_1))$, 
	it follows that
	$z$ does not belong to the invariant region 
	$\mathscr{R}_1(\mu,\sigma)$ at $t =t_1$ in general.

\end{remark}

\section{The proof of Lyapunov stability} 

Let us denote by $\mathscr{B}_{\varepsilon}$ 
the open ball with center at the origin and radius $\varepsilon$
in $\mathbb{R}^{2{\sf n} \times 2 {\sf n}}$.
In what follows,
we use the same letters as in the previous section
and assume that 
the average dwell time $\tau_a$ satisfies
\eqref{eq:adt_cond}.

The proof consists of three steps:
\begin{enumerate}
	\item
	Obtain an upper bound of
	the time $t_0$ at which the quantization process transitions
	from the ``zoom-out'' stage to the ``zoom-in'' stage.
	
	\item
	Show that there exists a time $t_{\varepsilon} \geq t_0$
	such that 
	the state $z$ satisfies $|z(t)| < \varepsilon$
	for all $t \geq t_{\varepsilon}$.
	
	\item
	Set $\delta > 0$ so that 
	if $|x(0)| < \delta$, then
	$|z(t)| < \varepsilon$
	for all $t < t_{\varepsilon}$.
\end{enumerate}
We break the proof of Lyapunov stability into the above three steps.

1)
Let $N \in \mathbb{N}$ satisfy \eqref{eq:N_sigma>0} and
let $\delta > 0$ be small enough to satisfy
\begin{equation}
	\label{eq:delta_cond_for_u=0}
	C_{\max}
	\cdot
	\Lambda^{N_0}
	\left( 
	\Lambda^{1/\tau_a}e^{\Gamma}
	\right)^{N\tau} \delta < \Delta_0.
\end{equation}
We see from the state bound \eqref{eq:state_bound} that 
$q_{\mu(t)}(y(t)) = 0$ for $t \in [0, N\tau]$
from Assumption \ref{ass:near origin}.
As we mentioned in Remark \ref{rem:brief_remark_for_Lyap} briefly,
Lemma \ref{lem:ADT_upperbound} implies that
the time $t_0$, at which the stage changes from ``zooming-out'' to ``zooming-in'',
satisfies $t_0 \leq N\tau$ for every switching signal
with the average dwell-time assumption \eqref{eq:ADT_cond}.

2) Fix $\alpha > 0$.
By \eqref{eq:xi_t0_tau},
$\xi(t_0) = 0$, and hence we see from \eqref{eq:mu_s_0_tau}
that $\mu(t_0)$ achieving 
$z(t_0) \in \mathscr{R}_1(\mu(t_0),\sigma(t_0))$
can be chosen so that
\begin{equation}
	\label{eq:bar_mu_ineq}
	\alpha \leq
	\mu(t_0) \leq
	\bar \mu,
\end{equation}
where $\bar \mu$ is defined by
\begin{align*}
	\bar \mu :=
	\max
	\Biggl\{\alpha,~~
	2\sqrt{\frac{\overline \lambda_P}{\underline \lambda_P}}
	\frac{ \Delta \tau C_{\max} 
		\Lambda^{N_0}
		\left(
		\Lambda^{1/\tau_a} e^{\Gamma}
		\right)^{(1+\chi) N\tau}}{M} \cdot
	\max_{p \in \mathcal{P}} 
	\left(
	\|W_p(\tau)^{-1}\|\Upsilon_p(\tau) \left\| e^{A_p\tau} \right\|
	\right)
	\Biggr\}.
\end{align*}
Note that $\bar \mu$ is independent of switching signals.

Let $\bar m > 0$ be the smallest integer satisfying
\begin{equation}
	\label{eq:m_cond_Lyapunov}
	\bar m > \frac{\log(\bar \mu M \sqrt{c^{N_0+T/\tau_a}}/
		(\varepsilon C_{\max}))}
	{\log (1/(\Omega\sqrt{c^{T/\tau_a}}))}.
\end{equation}
Define
$t_{\varepsilon} := t_0+\bar m T$.
Since $c \geq 1$ and
$\Omega \sqrt{c^{T/\tau_a}} < 1$, 
\eqref{eq:mu_update_non_switched} and
\eqref{eq:mu_update_switched} give
\begin{align*}
	\mu(t_{\varepsilon}+kT+t) 
	&= 
	\mu(t_0+(\bar m +k)T+t)) \\
	&\leq 
	\sqrt{c^{N_0+T/\tau_a}} \cdot
	\left(\Omega \sqrt{c^{T/\tau_a}} \right)^{\bar m+k} \mu(t_0) \\
	&\leq
	\sqrt{c^{N_0+T/\tau_a}} \cdot
	\left(\Omega \sqrt{c^{T/\tau_a}} \right)^{\bar m} \bar \mu
\end{align*}
for all $k \geq 0$ and $t \in [0,T)$. Since $\bar m$ satisfies
\eqref{eq:m_cond_Lyapunov}, it follows that that 
$\mathscr{R}_1(\mu(t),\sigma(t))$ lies in
$\mathscr{B}_{\varepsilon}$ for all $t \geq t_{\varepsilon}$.
Recall that 
$z(t_0) \in\mathscr{R}_1(\mu(t_0),\sigma(t_0))$ and that
$\mathscr{R}_1(\mu(t),\sigma(t))$ is
an invariant region for all $t \geq t_0$
from Theorem \ref{lem:convergence_to_origin}. Thus we have
\begin{equation}
	\label{eq:z_eps2}
	|z(t)| < \varepsilon \qquad (t \geq t_{\varepsilon}).
\end{equation}

3) Define
$
\underline{c} := \min \{
1, \min_{p_1,p_2 \in \mathcal{P},p_1 \not= p_2}c_{p_2,p_1}
\}.
$
Since $\underline c \leq 1$,
it follows from 
\eqref{eq:mu_update_non_switched}
\eqref{eq:mu_update_switched},
and \eqref{eq:bar_mu_ineq} that
\begin{align}
	\mu(t) 
	\geq \Omega^{\bar m}
	\sqrt{\underline{c}^{N_0+\bar m T/\tau_a}}\mu(t_0) 
	\geq 
	\alpha\Omega^{\bar m}
	\sqrt{\underline{c}^{N_0+\bar m T/\tau_a}}=: \eta
	\label{eq:mu_Omega_bound}.
\end{align}
for all $t \in [t_0, t_{\varepsilon}]$.
Set
$\delta > 0$ so that
\begin{gather}
	\label{eq:delta_cond1}
	C_{\max} \cdot
	\Lambda^{N_0}
	\left(
	\Lambda^{1/\tau_a} 
	e^{\Gamma }
	\right)^{N\tau +\bar m T} 
	\delta < 
	\eta \Delta_0\\
	\label{eq:delta_cond2}
	\Lambda^{N_0}
	\left(
	\Lambda^{1/\tau_a} 
	e^{\Gamma }
	\right)^{N\tau +\bar m T}  \delta < 
	\varepsilon/2.
\end{gather}

Since $t_{\varepsilon} = t_0+\bar m T \leq N\tau + \bar m T$, 
by \eqref{eq:state_bound}, \eqref{eq:delta_cond_for_u=0}, 
\eqref{eq:mu_Omega_bound}, and \eqref{eq:delta_cond1}, 
Assumption \ref{ass:near origin} gives
$q_{\mu(t)}(y(t))=0$ in the interval $[0,t_{\varepsilon}]$, so
$\xi(t) = 0$ and $u(t) = 0$ in the same interval.
Combining this with \eqref{eq:delta_cond2}, we obtain
$|x(t)| \leq 
\Lambda^{N_0}
\left(
\Lambda^{1/\tau_a}
e^{\Gamma}
\right)^{ (N\tau + \bar m T)} 
\delta < \varepsilon/2$
for all $t < t_{\varepsilon}$. 
Thus 
\begin{equation}
	\label{eq:z_eps1}
	|z(t)| = 2|x(t)| < \varepsilon \qquad (t < t_{\varepsilon}).
\end{equation}
From \eqref{eq:z_eps2} and \eqref{eq:z_eps1}, we see that 
Lyapunov stability can be achieved.
\hspace*{\fill} $\Box$
%

\section{Numerical examples}
Consider the continuous-time switched system \eqref{eq:ClosedSystem} with
the following two modes:
\begin{align*}
	(A_1,B_1,C_1)
	&=
	\left(
	\begin{bmatrix}
		1 & -0.3 \\ 0.4 & -4
	\end{bmatrix},\ 
	\begin{bmatrix}
		1 \\ 0 
	\end{bmatrix},\ 
	\begin{bmatrix}
		1 & 1
	\end{bmatrix}
	\right)\\
	(A_2,B_2,C_2) &=
	\left( 
	\begin{bmatrix}
		-0.1 & 1 \\ -1 & 0.1
	\end{bmatrix},\ 
	\begin{bmatrix}
		0 \\ 1
	\end{bmatrix},\ 
	\begin{bmatrix}
		0 & -1
	\end{bmatrix}
	\right)
\end{align*}
with jump matrices $R_{1,2} = R_{2,1} = I$.
As the feedback gain and the observer gain of each mode, we take
\begin{align*}
	(K_1,L_1) &=
	\left( 
	\begin{bmatrix}
		-3 & -2
	\end{bmatrix},~
	\begin{bmatrix}
		-4 \\ 0 
	\end{bmatrix}
	\right)\\
	(K_2,L_2) &=
	\left( 
	\begin{bmatrix}
		0 & 1
	\end{bmatrix},~
	L_2 = 
	\begin{bmatrix}
		0 \\ -1
	\end{bmatrix}
	\right).
\end{align*}
Let $q$ be a uniform-type quantizer
with parameters 
$
M = 10$,
$
\Delta = 0.05.
$
The parameters $\tau, \bar \tau, \chi$ in 
the 
``zooming-out'' stage are $\tau = 0.5$, $\bar \tau = 1$, 
and $\chi = 0.1$.
Also, define $Q_1$ and $Q_2$ in \eqref{eq:Lyapunov_cont} and $\kappa$ in 
\eqref{eq:M_Delta_cond} by
$
Q_1 := 
\diag(6,6,2,6)$,
$
Q_2 :=
\diag(1,1,1,1),
$
$\kappa := 4.5$,
where $\diag(e_1,\dots,e_4)$ means a diagonal matrix whose diagonal
elements starting in the upper left corner are $e_1,\dots,e_4$.
Then we obtain 
$T = 0.6025$ in \eqref{eq:T_def},
$\Omega = 0.9063$ in \eqref{eq:Omega_def}, 
$c = 1.9867$ in \eqref{eq:c_def_cont}, and
$\tau_a = 2.0744$ in \eqref{eq:adt_cond}.

Figure \ref{fig:cont_simulation} (a) and (b) show that 
the Euclidean norm of the state $x$ and the estimate $\xi$, and
the ``zoom'' parameter $\mu$,
respectively,
with initial condition $x(0) = [5~-\!10]^{\top}$ and $\mu(0) = 1$.
The vertical dashed-dotted line indicates the switching times
$t=3.5,7,20$.
In this example, the ``zooming-out'' stage finished at $t = 0.5$.
We see 
the non-smoothness of $x, \xi$ and
the increase of $\mu$ at the switching times 
$t=3.5,7,20$ because of
switches and quantizer updates.
Not surprisingly, the
adjustments of $\mu$ 
in \eqref{eq:mu_s_0_tau} and \eqref{eq:mu_update_switched}
are conservative.

\begin{figure}
	\centering
	\subcaptionbox{Norms of state $x$ and estimate $\xi$.}
	{\includegraphics[width = 9cm,clip]{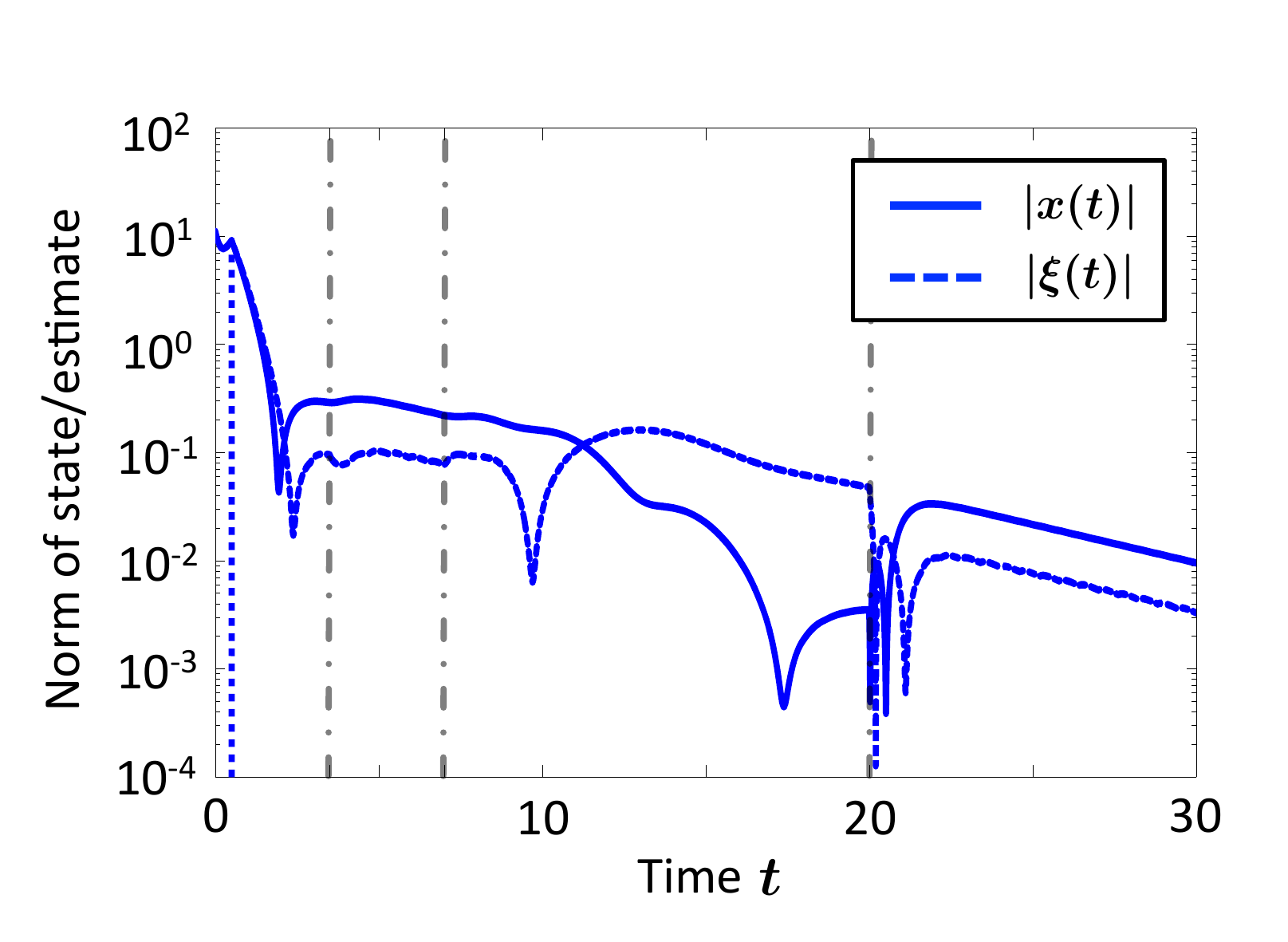}}
	
	\subcaptionbox{Zoom parameter $\mu$.}
	{\includegraphics[width = 9cm,clip]{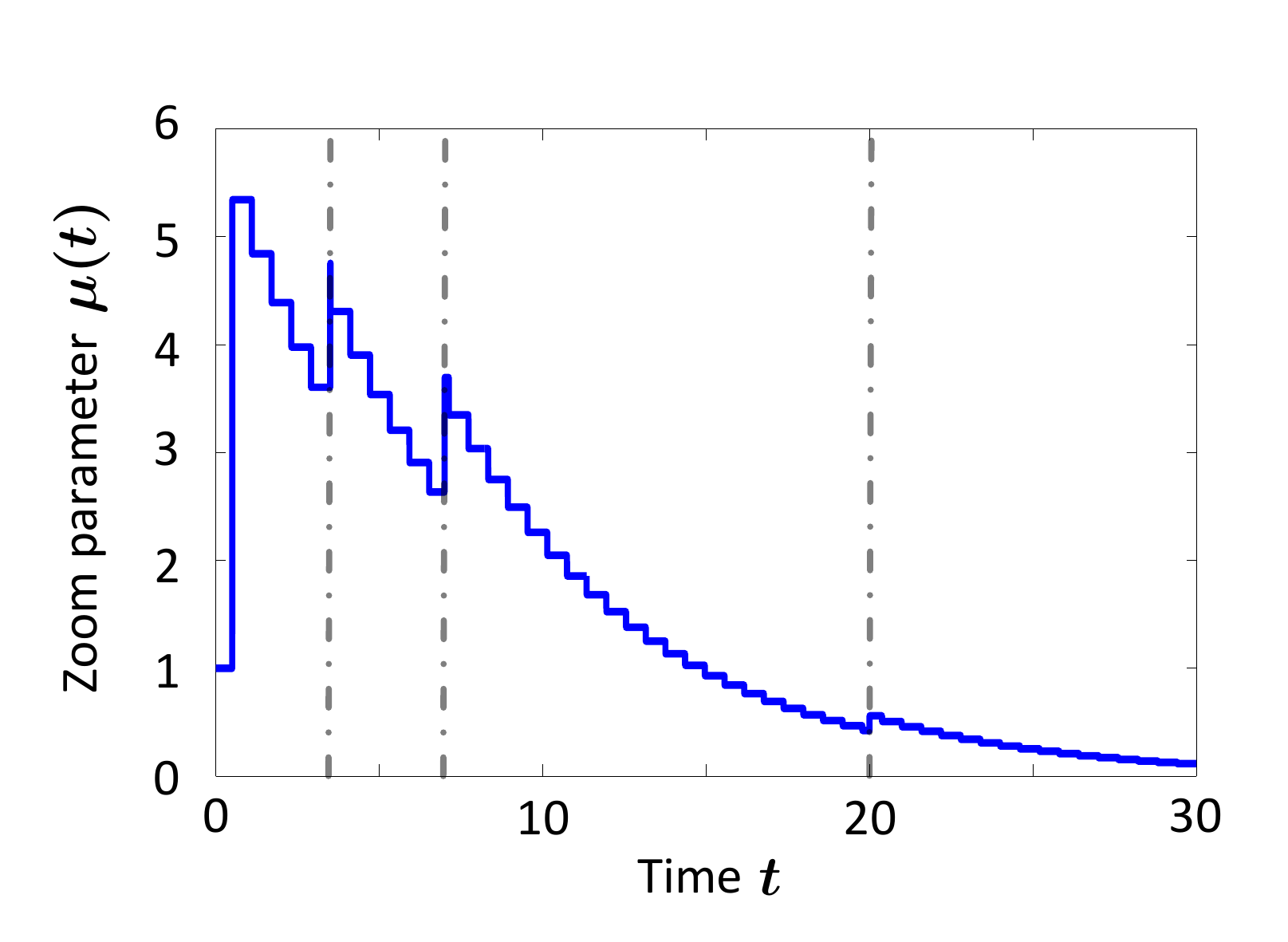}}
	\caption{Simulation with initial condition
		$x(0)= [5~-\!10]^{\top}$ and $\sigma(0) = 1$.
		The vertical dashed-dotted line indicates the switching times
		$t=3.5,7,20$.
		\label{fig:cont_simulation}}
\end{figure}

\section{Concluding remarks}
We have proposed an update rule of dynamic quantizers 
to stabilize continuous-time switched systems with quantized output
feedback.
The average dwell-time property has been utilized
for the state reconstruction in the 
``zooming-out'' stage and 
for convergence to the origin in the ``zooming-in'' stage.
The update rule not only periodically decreases the ``zoom'' parameter
to drive the state to the origin, but also adjusts the parameter
at each switching time
to avoid quantizer saturation.
Future work involves 
designing the controller and 
the quantizer simultaneously,
and addressing more general systems by incorporating
disturbances and nonlinear dynamics.

	\ifCLASSOPTIONcaptionsoff
	\newpage
	\fi
	


\begin{thebibliography}{10}
	\providecommand{\url}[1]{#1}
	\csname url@samestyle\endcsname
	\providecommand{\newblock}{\relax}
	\providecommand{\bibinfo}[2]{#2}
	\providecommand{\BIBentrySTDinterwordspacing}{\spaceskip=0pt\relax}
	\providecommand{\BIBentryALTinterwordstretchfactor}{4}
	\providecommand{\BIBentryALTinterwordspacing}{\spaceskip=\fontdimen2\font plus
		\BIBentryALTinterwordstretchfactor\fontdimen3\font minus
		\fontdimen4\font\relax}
	\providecommand{\BIBforeignlanguage}[2]{{%
			\expandafter\ifx\csname l@#1\endcsname\relax
			\typeout{** WARNING: IEEEtran.bst: No hyphenation pattern has been}%
			\typeout{** loaded for the language `#1'. Using the pattern for}%
			\typeout{** the default language instead.}%
			\else
			\language=\csname l@#1\endcsname
			\fi
			#2}}
	\providecommand{\BIBdecl}{\relax}
	\BIBdecl
	
	\bibitem{Nair2007}
	G.~N. Nair, F.~Fagnani, S.~Zampieri, and R.~J. Evans, ``Feedback control under
	data rate constraints: An overview,'' \emph{Proc. IEEE}, vol.~95, pp.
	108--137, 2007.
	
	\bibitem{Ishii2012}
	H.~Ishii and K.~Tsumura, ``{Data rate limitations in feedback control over
		network},'' \emph{IEICE Trans. Fundamentals}, vol. E95-A, pp. 680--690, 2012.
	
	\bibitem{Deaecto2010}
	G.~S. Deaecto, J.~C. Geromel, F.~S. Garcia, and J.~A. Pomilio, ``{Switched
		affine systems control design with application to DC--DC converters},''
	\emph{IET Control Theory Appl.}, vol.~4, pp. 1201--1210, 2010.
	
	\bibitem{Rinehart2008}
	M.~Rinehart, M.~Dahleh, D.~Reed, and I.~Kolmanovsky, ``{Suboptimal control of
		switched systems with an application to the DISC engine},'' \emph{IEEE Trans.
		Control Systems Tech.}, vol.~16, pp. 189--201, 2008.
	
	\bibitem{Shorten2007}
	R.~Shorten, F.~Wirth, O.~Mason, K.~Wulff, and C.~King, ``{Stability criteria
		for switched and hybrid systems},'' \emph{SIAM Review}, vol.~49, pp.
	545--592, 2007.
	
	\bibitem{Lin2009}
	H.~Lin and P.~J. Antsaklis, ``{Stability and stabilizability of switched linear
		systems: a survey of recent results},'' \emph{IEEE Trans. Automat. Control},
	vol.~54, pp. 308--322, 2009.
	
	\bibitem{Liberzon2003Book}
	D.~Liberzon, \emph{Switching in Systems and Control}.\hskip 1em plus 0.5em
	minus 0.4em\relax Birkh\"auser, Boston, 2003.
	
	\bibitem{Nair2003}
	G.~N. Nair, S.~Dey, and R.~J. Evans, ``{Infinmum data rates for stabilising
		Markov jump linear systems},'' in \emph{Proc. 42nd IEEE CDC}, 2003.
	
	\bibitem{Xiao2010}
	N.~Xiao, L.~Xie, and M.~Fu, ``{Stabilization of Markov jump linear systems
		using quantized state feedback},'' \emph{Automatica}, vol.~46, pp.
	1696--1702, 2010.
	
	\bibitem{Xu2013}
	Q.~Xu, C.~Zhang, and G.~E. Dullerud, ``{Stabilization of Markovian jump linear
		systems with log-quantized feedback},'' \emph{J. Dynamic Systems, Meas,
		Control}, vol. 136, pp. 1--10 (031\,919), 2013.
	
	\bibitem{Wakaiki2014IFAC}
	M.~Wakaiki and Y.~Yamamoto, ``{Quantized feedback stabilization of sampled-data
		switched linear systems},'' in \emph{Proc. 19th IFAC WC}, 2014.
	
	\bibitem{Brockett2000}
	R.~W. Brockett and D.~Liberzon, ``Quantized feedback stabilization of linear
	systems,'' \emph{IEEE Trans. Automat. Control}, vol.~45, pp. 1279--1289,
	2000.
	
	\bibitem{Liberzon2003Automatica}
	D.~Liberzon, ``Hybrid feedback stabilization of systems with quantized
	signals,'' \emph{Automatica}, vol.~39, pp. 1543--1554, 2003.
	
	\bibitem{Hespanha1999CDC}
	J.~P. Hespanha and A.~S. Morse, ``Stability of switched systems with average
	dwell-time,'' in \emph{Proc. 38th IEEE CDC}, 1999.
	
	\bibitem{Liberzon2014}
	D.~Liberzon, ``Finite data-rate feedback stabilization of switched and hybrid
	linear systems,'' \emph{Automatica}, vol.~50, pp. 409--420, 2014.
	
	\bibitem{Wakaiki2014CDC}
	M.~Wakaiki and Y.~Yamamoto, ``{Output feedback stabilization of switched linear
		systems with limited information},'' in \emph{Proc. 53rd IEEE CDC}, 2014.
	
	\bibitem{Yang2015ACC}
	G.~Yang and D.~Liberzon, ``Stabilizing a switched linear system with
	disturbance by sampled-data quantized feedback,'' in \emph{Proc. ACC'15},
	2015.
	
	\bibitem{Liberzon2003}
	D.~Liberzon, ``On stabilization of linear systems with limited information,''
	\emph{IEEE Trans. Automat. Control}, vol.~48, pp. 304--307, 2003.
	
	\bibitem{WakaikiMTNS2014}
	M.~Wakaiki and Y.~Yamamoto, ``{Quantized output feedback stabilization of
		switched linear systems},'' in \emph{Proc. MTNS'14}, 2014.
	
	\bibitem{Liberzon2007}
	D.~Liberzon and D.~Ne\v{s}i\'{c}, ``Input-to-state stabilization of linear
	systems with quantized state measurement,'' \emph{IEEE Trans. Automat.
		Control}, vol.~52, pp. 767--781, 2007.
	
	\bibitem{Filippov1988}
	A.~F. Filippov, \emph{Differential Equations with Discontinuous Righthand
		Sides}.\hskip 1em plus 0.5em minus 0.4em\relax Dordrecht: Kluwer, 1988.
	
	\bibitem{Tanwani2011}
	A.~Tanwani and D.~Liberzon, ``{Robust invertibility of switched linear
		systems},'' in \emph{Proc. 50th CDC}, 2011.
	
	\bibitem{Ma2015}
	D.~Ma and J.~Zhao, ``{Stabilization of networked switched linear systems: An
		asynchronous switching delay system approach},'' \emph{Systems Control
		Lett.}, vol.~77, pp. 46--54, 2015.
	
\end{thebibliography}
\end{document}